\documentclass[11pt]{amsart}

\usepackage{a4wide}

\usepackage{harvard}

\usepackage{amsfonts}

\usepackage{color}

\newtheorem{lem}{Lemma}
\newtheorem{thm}{Theorem}

\newtheorem{cor}{Corollary}
\newtheorem{prop}{Proposition}
\newtheorem{defi}{Definition}

\newcommand{\ci}{\mathcal{I}}
\newcommand{\cj}{\mathcal{J}}
\newcommand{\cn}{\mathcal{N}}

\newcommand{\cs}{\mathcal{S}}

\newcommand{\bro}{\rho_0^{\cj}}

\newcommand{\R}{\mathbb{R}}
\renewcommand{\O}{\mathbb{O}}

\newcommand{\sti}[2]{\genfrac{}{}{0pt}{}{#1}{#2}}

\title[Global stability for a class of virus models]{Global stability for a class of virus models with CTL immune response and antigenic variation}
\author{Max O. Souza}
\address{%
Departamento de Matem\'atica Aplicada, Universidade Federal
Fluminense, R. M\'ario Santos Braga, s/n, Niter\'oi, RJ 22240-120, Brazil.
}

\email{msouza@mat.uff.br}

\author{Jorge P. Zubelli}
\address{%
IMPA, Est. D. Castorina 110, Rio de Janeiro, RJ 22460-320, Brazil
}

\email{zubelli@impa.br}

\date{\today}

\thanks{ The authors thank an anonymous referee for suggesting that an early version of  Theorem  1 could be extended as presented here. MOS is partially supported by FAPERJ grants  170.382/2006 and  110.174/2009. JPZ was supported by CNPq under grants 302161/2003-1 and 474085/2003-1.}
 


\begin{document}

\begin{abstract}
We study the global stability of a class of models for in-vivo virus dynamics, that take into account the CTL  immune response and display antigenic variation. This class includes a number of models that have been extensively used to model HIV dynamics. We show that  models in this class are globally asymptotically stable, under mild hypothesis, by using appropriate Lyapunov functions. We also characterise the  stable equilibrium points for the entire biologically relevant parameter range. As a byproduct, we are able to determine  what is  the diversity of the persistent  strains.
\end{abstract}

\maketitle

\section{Introduction}

\subsection{Models for in-vivo virus dynamics}
A number of population dynamics models have been proposed in order to describe the HIV in-vivo dynamics \cite{Perelson:Nelson:1999,Nowak:May:2000}. Although these models have distinct features, since they attempt  to incorporate different aspects of the interaction between the virus and the  immune system, many of them share a common long-term behaviour, evolving towards an isolated equilibrium state \cite{Nowak:May:2000}.

The basic model for the HIV in-vivo dynamics is given by a three-by-three, first-order system of ordinary differential equations (ODEs)---\cite{Nowak:Bangham:1996,Bonhoeffer:May:Nowak:1997,Nowak:May:2000}:
\begin{equation}
 \left\{
\begin{array}{rcl}
 \dot{x}&=&\lambda-dx-\beta xv,\\
\dot{y}&=&\beta x v - ay,\\
\dot{v}&=& ky-uv.\\
\end{array}
\right.
\label{nb:basic}
\end{equation}
In this model, $x$ denotes the uninfected cells, $y$ the infected cells and $v$ the free virus particles. The average lifetime of an infected cell is $1/a$, while the average lifetime of a virus particle is $1/u$. The total number of virus particles  produced by an infected cell is $k/a$. Healthy cells are infected at a rate $\beta xv$.  New CD4+ T cells are produced, in the thymus, at a rate $\lambda$, and die at a rate $dx$.

System \eqref{nb:basic} has two equilibrium points:
\begin{enumerate}
 \item The disease free equilibrium:  $x^*=\lambda/d$, $y^*=v^*=0$;
\item The endemic equilibrium: $x^*=au/\beta k$, $y^*=(\beta\lambda k - dau)/\beta a k$, $v^*=(\beta\lambda k - dau)/\beta a u$.
\end{enumerate}
The long term dynamics of System~\eqref{nb:basic} can be entirely described in terms of the dimensionless
 parameter
\begin{equation}
 R_0=\frac{\beta\lambda k}{dau},
\label{r0:def}
\end{equation}
also known as the \textit{basic reproductive ratio}. 

If $R_0\le 1$, the disease free equilibrium is a global attractor, and the infection cannot persist. If $R_0>1$, the endemic equilibrium becomes a global attractor, and the infections persists indefinitely. This has been first observed numerically ~\cite{Nowak:Bangham:1996,Bonhoeffer:May:Nowak:1997,Nowak:May:2000}. Mathematical proofs of these global stability characteristics were given by \citeasnoun{Li:Muldowney:1995}, using Hirsch's theory of competitive differential systems---see \citeasnoun{Smith:1995}---and, more recently  by \citeasnoun{Leenheer:Smith:2003} and \citeasnoun{Korobeinikov:2004} using a Lyapunov function approach.
 
Given the notable ability of the HIV to escape from immune response, there is interest in studying models that account for a more detailed immune response as, for instance, the role of cytotoxic T lymphocytes (CTLs). An example is the following four-by-four system of ODEs \cite{Nowak:Bangham:1996}:
\begin{equation}
 \left\{
\begin{array}{rcl}
 \dot{x}&=&\lambda-dx-\beta xv,\\
\dot{y}&=&\beta x v - ay - pyz,\\
\dot{v}&=& ky-uv,\\
\dot{z}&=&cyz-bz.
\end{array}
\right.
\label{nb:immune}
\end{equation}
System \eqref{nb:immune} extends \eqref{nb:basic} by introducing the $z$ variable, that denotes the CTL response. Infected cells are killed at a rate $pyz$, while antigen stimulation produces CTL cells at a rate $cyz$. In the absence of such a stimulation, CTL cells decay at a rate $bz$.

In the same vein, the high mutation rate of HIV naturally leads  to the study of the interplay between immune response and virus diversity for a number of different strains. The immune response produces a selection pressure on these different strains of the virus, as discussed in \citeasnoun{Nowak:Bangham:1996}, when studying and  numerically analysing a $(3n+1)$-by-$(3n+1)$ first-order ODE system of the form:
\begin{equation}
 \left\{
\begin{array}{rclr}
  \dot{x}&=&\lambda-dx-x\sum_{i=1}^n\beta_i v_i,& \\
\dot{y_i}&=&\beta_i x v_i - a_iy_i - p_iy_iz_i,& i=1,\ldots,n,\\
\dot{v_i}&=& k_iy_i-u_iv_i, & i=1,\ldots,n,\\
\dot{z_i}&=&c_iy_iz_i-b_iz_i, & i=1,\ldots,n.
\end{array}
\right.
\label{nb:antigenic}
\end{equation}
System \eqref{nb:antigenic} is a slightly generalised form of the system studied by \citeasnoun{Nowak:Bangham:1996}, where  the more restricted case $a_i=a$, $p_i=p$, $u_i=u$, $c_i=c$, and $b_i=b$ was addressed.  

In all these models, there is the question as whether the  long term dynamics approaches an equilibrium or, more generally, an attractor, and how this might depend on the initial condition. There is compelling numerical evidence---cf. \cite{Nowak:Bangham:1996,Bonhoeffer:May:Nowak:1997,Nowak:May:2000}---that Systems \eqref{nb:immune} and \eqref{nb:antigenic} are globally asymptotically stable. However, no mathematical proof of this fact seems to be available. In System~\eqref{nb:antigenic} there is also the question of determining the antigenic diversity at  the equilibrium.

In this work, we study the stability characteristics of the models given by \eqref{nb:antigenic} following a Lyapunov approach. The Lyapunov functional used here has been used before by \citeasnoun{Korobeinikov:Wake:1999}, in the global analysis of three-dimensional predator-prey systems, and by \citeasnoun{Korobeinikov:2004} and \citeasnoun{Korobeinikov:2004b} in the global analysis of various virus models. More precisely, by using an appropriate linear combination of this Lyapunov functional,  we are able to show global asymptotic stability results for System \eqref{nb:antigenic}, and hence to \eqref{nb:immune}. 
%%%%%%%%%%%%%%%%%%%%JPZ CHANGE
% In Section 2, we study the global stability characteristics of \eqref{nb:immune}. In Section 3, we study the equilibria and global stability of System \eqref{nb:antigenic} under the assumption of unique fitnesses of the strains. In this case,  we determine the possible equilibria of \eqref{nb:antigenic}, and show than there are $2^{n-1}(n+2)$ equilibrium points. We also show that system is globally asymptotic stable, and we determine what is the global attractor in the nonnegative orthant of $\R^{3n+1}$. As a byproduct, we characterise the attained diversity and show that it is monotonically increasing with the strength of the immune response. Some additional results for the case of nonunique fitnesses are also presented.

The plan for this article goes as follows: We close this introductory section with further biological background and motivations. In Section~\ref{prelims} we address some preliminary issues such as choice of dimensionless variables and parameter reductions.
We study the global stability characteristics of model~\eqref{nb:immune}.
In Section~\ref{antimodel}, we study the equilibria and global stability of model~\eqref{nb:antigenic} under the assumption of unique fitnesses of the strains. In this case,  we determine the possible equilibria of \eqref{nb:antigenic}, and show than there are $2^{n-1}(n+2)$ equilibrium points. We also show that system is globally asymptotic stable, and we determine what is the global attractor in the nonnegative orthant of $\R^{3n+1}$. As a byproduct, we characterise the attained diversity and show that it is monotonically increasing with the strength of the immune response. Some additional results for the case of nonunique fitnesses are also presented.
We conclude in Section~\ref{conclusions} with a discussion of some of the implications
of our results.

\subsection{Biological Background and Motivation}

The models studied by equation~\ref{nb:antigenic} have many potential biological applications. Most notably, to within-host infections connected to
cytotoxic T lymphocytes with antigenic variation including, but not restricted to, HIV infection.
A better understanding of how the within-host HIV, interacts with immune cells seems
to be a key factor in the development of effective long-term therapies or possibly preventive vaccines
for deadly diseases such as the acquired immunodeficiency syndrome \cite{Nowak:May:2000}. 
Mathematical modeling of the underlying biological mechanisms and a good understanding of the theoretical implications of such models is crucial in this process. 
Indeed, it helps clarifying and testing assumptions, finding the smallest number of determining factors to explain the biological phenomena, and analysing the experimental results \cite{AB2003}.
Furthermore, modeling has already impacted on research at molecular level \cite{Nowak:May:2000} and  important
results have been obtained in modeling the virus dynamics for several infections, such as the 
HIV \cite{Nowak:Bangham:1996,PKD1993,PNML1996}, hepatitis B \cite{MRB1991}, hepatitis C \cite{NLDG1998}, and influenza \cite{BR1994a}.

%We start by recalling the path followed by the within-host HIV infection~\citep{NM}. 
In the particular case of the HIV infection, the dynamics of the within-host infection goes as follows: First, the HIV enters a T cell. 
Being a retrovirus, once the HIV is  
inside the T~cell, it makes a DNA copy of its viral RNA. For this process it requires the reverse transcriptase (RT) enzyme. The DNA of the virus is then inserted in the T-cell's DNA. The latter in turn will produce viral particles that can bud off the T~cell to infect other ones.  Before one such viral particle leaves 
the infected cell, it must be equipped with {\em protease}, which is an enzyme used to cleave a long
protein chain. Without protease the virus particle is uncapable of infecting other T cells. 

One of the key characteristics of HIV is its extensive genetic variability. In fact, the HIV seems
to be changing continuously in the course of each infection and typically the virus strain that initiates the patient's infection differs from the one found a year ore more after the infection. In this respect, the introduction of the different strains in the model
is crucial for it to be realistic.

In general terms, one can also say that Model (\ref{nb:antigenic}) is similar in spirit to other models such as food-chain models. The latter have attracted substantial interest by a number of authors. See for example~\cite{roysolimano,kooibsb2001,kooibsb1998} and references therein. 
However, the presence of more general quadratic terms or of logistic terms on the
right hand side of the different ``strains" leads to a potentially richer dynamics than
the globally stable present in Model \ref{nb:antigenic}. See for example \cite{kooibsb2001} for a bifurcation analysis of certain food chain models.

\section{Preliminaries} \label{prelims}

\subsection{Parameter reduction}

As already noticed in the introduction, a more restricted form of \eqref{nb:antigenic} with a number of parameters being strain-independent has been studied by \citeasnoun{Nowak:Bangham:1996}. It turns out that some parameters in \eqref{nb:antigenic} can indeed be taken to be strain-independent, as we now show.

We start by noting that if $k_i=0$, then $v_i$ decays exponentially with rate $u_i$. Also, if $p_i=0$, then the dynamics of $z_i$ does not impinges on the rest of the system. Thus, without loss of generality, we can assume that $k_i,p_i\not=0$, $i=1,\ldots,n$. In this case, following \citeasnoun{Pastore:2005}, we rescale the $v_i$s and $\beta_i$s. In addition, we also rescale the $z_i$s. More precisely, the change of variables
\begin{equation}
 v_i\mapsto \frac{k_i}{k}v_i, \quad
z_i\mapsto \frac{p}{p_i}z_i
\quad\text{and}\quad
\beta_i \mapsto \frac{k}{k_i}\beta_i
\label{red:change}
\end{equation}
takes \eqref{nb:antigenic} into
\begin{equation}
 \left\{
\begin{array}{rclr}
\dot{x}&=&\lambda-dx-x\sum_{i=1}^n\beta_i v_i,& \\
\dot{y_i}&=&\beta_i x v_i - a_iy_i - py_iz_i,& i=1,\ldots,n,\\
\dot{v_i}&=& ky_i-u_iv_i, & i=1,\ldots,n,\\
\dot{z_i}&=&c_iy_iz_i-b_iz_i, & i=1,\ldots,n.
\end{array}
\right.
\label{our:antigenic}
\end{equation}
Intuitively, the change of variables \eqref{red:change} reflects that only the ratio $\beta_i/k_i$ turns out to be important, and that this can be already taken into account in the $\beta_i$s, provided we rescale the $v_i$s. Moreover, it also shows that the precise value $p_i$ does not matter, as long as it is nonzero. 

\subsection{Dimensionless constants}

In \citeasnoun{Nowak:Bangham:1996}, it was already  observed that, in addition to $R_0$, the quantity $cy^*/b$ is also important in determining the global equilibria.  In a more precise fashion, \citeasnoun{Nowak:May:2000} define 
\[
 R_{I}=1+\frac{\beta b k}{c d u},
\]
which they term the basic reproductive ratio in the presence of immune response. 
\
However, we follow \citeasnoun{Pastore:2005}, and find  more convenient to write %look whether the ratio 
\[
 R_{I}=1+\frac{R_0}{I_0},
\]
where
\[
 I_0=\frac{c\lambda}{ab}.
\]
An alternative dimensionless constant is the CTL reproduction number given by
\[
 P_0=\frac{I_0(R_0-1)}{R_0}.%,\quad (R_0-1)^+=\max\{R_0-1,0\}.
\]
Although only two constants among $R_0$, $I_0$ and $P_0$ are independent, and are sufficient to describe the regimes of \eqref{our:antigenic}, we have chosen to use both, three constants, as some conditions are better characterised by $P_0$, while much of the algebra in the Lyapunov functional derivatives is better handled by expressing them in terms of $R_0$ and $I0$. Thus, we shall use the strain dependent constants:
%
% The strong responders, see \S 2.3 below, are easily characterised by the condition $P_0^i>1$. However, expression for the derivative of the Lyapunov functionals turn out to be more cumbersome in this notation. Thus, we have chosen to present the results using  following strain nondimensional numbers
%
\begin{equation}
 R_0^i=\frac{\beta_i\lambda k}{da_iu_i},\quad I_0^i=\frac{c_i\lambda}{a_ib_i}\quad
\text{and}\quad P_0^i=\frac{I_0^i(R_0^i-1)}{R_0^i}.
\end{equation}
\subsection{Strain sets}

In order to deal with plethora of equilibria that arises in System~\eqref{our:antigenic}, we shall now define some notation for some special set of strain indices. This will allows us to deal conveniently with the combinatorial structure of the equilibria.

Without loss of generality, we shall assume that the strains are indexed by increasing order of $R_0^i$.

Let $\cn=\{1,2,\ldots,n\}$. Then, we define the set of the strong responders as
\[
 \cs=\{i \in\cn \,| \,P_0^i>1\}.
\]
\begin{defi}
 We shall say that the set $\cs$ of strong responders is consistent, if
\[
 \cs=\{1,\ldots,m\}, \quad 1\leq m\leq n.
\]
This is certainly the cases, if the $I_0^i$ satisfy $I_0^{i}\geq I_0^{i+1}$. In particular, this holds if $I_0^i=I_0$ as in the model studied by \citeasnoun{Nowak:Bangham:1996}.
\end{defi}
Given a set of indices $\ci$, we define
\[
 \rho_0^\ci=\sum_{i\in\ci}\frac{R_0^i}{I_0^i}.
\]
Two important definitions are given below:
\begin{defi}
We shall say that $\ci\subset\cs$ is an antigenic set, if 
\begin{equation}
 R_0^i\geq 1 + \rho_0^\ci,\quad i\in\ci 
\label{ineq:cond:1}
\end{equation}
holds. In addition, if
\begin{equation}
 R_0^i\leq 1 + \rho_0^\ci,\, i\not\in\ci.\label{ineq:cond:2}
\end{equation}
also holds, we shall say that $\ci$ is a stable antigenic set.

Notice that, if $\cs\not=\emptyset$, we have that $\ci=\{1\}$ is an antigenic set. Let $l$ be the largest integer for which the set $\cj=\{1,\ldots,l\}$ is an antigenic set. Then we shall say that $\cj$ is the maximal antigenic set. 
\end{defi}
Two important facts about the maximal and stable antigenic sets are collected below:
\begin{lem}
\label{lem:sets}
Assume that $\cs\not=\emptyset$, and that the strain basic reproductive numbers are distinct.
\begin{enumerate}
\item 
If a stable antigenic set exists, then it  is also the  maximal antigenic set. In particular, stable antigenic sets are unique.
 \item If $\cn$ is the maximal antigenic set, then it is a stable antigenic set.
\end{enumerate}
\end{lem}
\begin{proof}
\begin{enumerate}

\item Assume that $\ci$ is a stable antigenic set and let  $i_l=\max\ci$. If $1\leq k<l$, then $i_k\in\ci$. Indeed, by the increasing ordering and \eqref{ineq:cond:1}, we have that
\[
 R_0^{i_k}> R_0^{i_l}\geq 1+\rho_0^\ci.
\]
But this contradicts \eqref{ineq:cond:2}, thus we must have $i_k\in\ci$.
Now assume that $\ci'=\{i_1,\ldots,i_{l+1}\}$ is also an antigenic set. Then we must have
\[
 R_0^{l+1}\geq 1 + \rho_0^{\ci'}= 1 + \rho_0^\ci+\frac{R_0^{l+1}}{I_0^{l+1}}\geq 1+\rho^\ci.
\]
But this again contradicts \eqref{ineq:cond:2} and, hence, that $\ci$ is stable antigenic. Therefore, it is maximal.
\item This follows since, in this case, \eqref{ineq:cond:2} cannot be violated.
\end{enumerate}
\end{proof}

\section{The model with antigenic variation}\label{antimodel}

In this section we shall study the stability of of \eqref{our:antigenic} in the non-negative orthant of $\R^{3n+1}$ which we shall denote by $\O$. The positive orthant will be denoted by $\O^+$.

We observe that the planes $z_i=0$ and that $\O$ %the nonnegative orthant of $\R^{3n+1}$
are positive invariant sets for \eqref{nb:antigenic}, since the field points inwards.

The equilibria and stability characteristics of \eqref{our:antigenic} depend  significantly whether the $R_0^i$s are distinct or not. In \S 3.1 and \S 3.2, we describe the equilibria and study their stability in the case of unique fitnesses, i.e., we assume that if $R_0^i=R_0^j$, then $i=j$. In this case, with the adopted order, we have  that 
\[
 R_0^{i}>R_0^{i+1},\quad i=1,\ldots,n-1.
\]
Additional remarks when the fitnesses are not unique can be found in Section 3.3.

\subsection{Equilibria}

 Let $\cn=\{1,2,\ldots,n\}$.  It turns out that the equilibria of \eqref{nb:antigenic} can be conveniently indexed by $(j,\cj)$, where $\cj\subseteq\cn$, and either $j=0$ or $j\not\in\cj$. The corresponding equilibrium point will be denoted by $X_{j,\cj}$. 

Using this notation, we have

\begin{lem} 
\label{lem:eq}
System \eqref{our:antigenic} has $2^{n-1}(2+n)$ equilibrium points which can be written as
\[
 X_{j,\cj} = \left(\frac{\lambda}{d}Q_{j,\cj}^x,\frac{\lambda}{a_1}Q_{j,\cj}^{y_1},\ldots,\frac{\lambda}{a_n}Q_{j,\cj}^{y_n}
,\frac{d}{\beta_1}Q_{j,\cj}^{v_1},\ldots,\frac{d}{\beta_n}Q_{j,\cj}^{v_n},\frac{a_1}{p}Q_{j,\cj}^{z_1},
\ldots,\frac{a_n}{p}Q_{j,\cj}^{z_n}\right),
\]
where
\begin{enumerate}
 \item 
\[
Q_{0,\emptyset}^x=1,\quad\text{and}\quad Q_{0,\emptyset}^{y_i}=Q_{0,\emptyset}^{v_i}=Q_{0,\emptyset}^{z_i}=0.
\]
\item If $j$ is  such that $1\leq j \leq n$, then we have
\[
Q_{j,\emptyset}^x=\frac{1}{R_0^j},\quad Q_{j,\emptyset}^{y_j}=1-\frac{1}{R_0^j},\quad 
Q_{j,\emptyset}^{v_j}=R_0^j-1\quad\text{and}\quad Q_{j,\emptyset}^{z_j}=0.
\]
and
 \[
Q_{j,\emptyset}^{y_i}=Q_{j,\emptyset}^{v_i}=Q_{j,\emptyset}^{z_i}=0,\quad  i=1,\ldots,n,\quad i\not=j.
\]
\item Given $\cj\subseteq \cn$, we have
\[
 Q_{0,\cj}^x=\frac{1}{1+\bro}
\]
and
\[
 Q_{0,\cj}^{y_i}=\frac{1}{I_0^i},\quad Q_{0,\cj}^{v_i}=\frac{R_0^i}{I_0^i},\quad Q_{0,\cj}^{z_i}=\frac{R_0^i}{1+\bro}-1, \quad i\in\cj;
\]
also 
\[
 Q_{0,\cj}^{y_i}= Q_{0,\cj}^{v_i}= Q_{0,\cj}^{z_i}=0,\quad i\not\in\cj.
\]
\item Given a proper subset $\cj\subset \cn$, and  $1\leq j'\leq n, j'\not\in \cj$, we have that
\[
 Q_{j',\cj}^x=\frac{1}{R_0^{j'}},\quad Q_{j',\cj}^{y_{j'}}=1-\frac{1}{R_0^{j'}}-\frac{\bro}{R_0^{j'}},\quad
Q_{j',\cj}^{v_{j'}}=R_0^{j'}-1-\bro,\quad Q_{j',\cj}^{z_{j'}}=0;
\]

for $i\in \cj$, we have
\[
 Q_{j',\cj}^{y_i}=\frac{1}{I_0^i},\quad Q_{j',\cj}^{v_i}=\frac{R_0^i}{I_0^i}, \quad 
Q_{j',\cj}^{z_i}=\frac{R_0^i}{R_0^{j'}}-1.
\]
For $i\not\in \cj$, and $i\not=j'$, we have
\[
 Q_{j',\cj}^{y_i}=Q_{j',\cj}^{v_i}=Q_{j',\cj}^{z_i}=0.
\]
\end{enumerate}
\end{lem}

\begin{proof}
 The first equilibrium is trivial. The second type of equilibria is obtained by choosing an index $j$ such that $z_j=0$, but $y_j\not=0$. We can choose only one such $j$, since this determines $x$. For the other indices $i$, we set $y_i=v_i=z_i=0$. The first equation then determines $v_j$. The third type is obtained by choosing a set $\cj$ of indices, such that, for $i\in\cj$, we have $z_i\not=0$. This readily determines $y_i$ and $v_i$. For $i\not\in\cj$, we have $y_i=v_i=z_i=0$. The first equation, then, determines $x$. Finally, the last equilibria is found by having a set of indices $\cj$, as in the equilibrium of the third type, and then choosing an index $j'\not\in\cj$ as in the second equilibrium. Again, only one such $j'$ can be chosen. 
\end{proof}

\subsection{Stability analysis}

We are now ready to study the global stability of the equilibria of System~\ref{our:antigenic}. Surprisingly, although there is a large number of equilibria, only four of them will be globally stable. In what follows, unless otherwise is said, we shall assume that that $R_0^{i}>R_0^{i+1}$, for $i=1,\ldots,n-1$, and that the set of strong responders is consistent.

\begin{thm}
 For system \eqref{our:antigenic}, defined on $\O$, and  with initial condition at its interior, there is always a globally asymptotically stable equilibrium 
given as follows:
\begin{enumerate}
 \item $X_{0,\emptyset}$, if $R_0^n\leq 1$; 
\item $X_{1,\emptyset}$, if $1<R_0^1$, and $P_0^1\leq1$. 
\item If $P_0^1>1$, let $\cj$ be the maximal antigenic set. Then
\begin{enumerate}
\item If $\cj$ is a stable antigenic set, then the equilibrium $X_{0,\cj}$ is globally asymptotically stable.
\item  Otherwise, let $j'$ be the smallest integer such that 
$j'\not\in\cj$, which exists by virtue of Lemma~\ref{lem:sets}. Then the equilibrium $X_{j',\cj}$ is globally asymptotically stable.
\end{enumerate}
\end{enumerate}

\label{thm:mutation}
\end{thm}

\begin{proof}[Proof of Theorem~\ref{thm:mutation}]

Following \citeasnoun{Korobeinikov:2004}, we shall use the following Lyapunov function:
\begin{align*}
 V(x,\mathbf{y},\mathbf{v},\mathbf{z})&=x-x^*\ln(x/x^*) + \sum_{i=1}^n\left(y_i-y_i^*\ln(y_i/y_i^*)\right) +\\
& \qquad + \sum_{i=1}^nC_i\left(v_i-v_i^*\ln(v_i/v_i^*)\right)
+ p\sum_{i=1}^n\frac{1}{c_i}\left(z_i-z_i^*\ln(z_i/z^*)\right),
\end{align*}
where $C_i$ will be a constant to be specified later on. 

Then, using the uniform notation of the the equilibria of \eqref{our:antigenic}, that is set  in Lemma~\ref{lem:eq}, see \S3.1, we have that
\begin{align}
 \dot{V}=& \frac{d}{dt}V(x(t),\mathbf{y}(t),\mathbf{v}(t),\mathbf{z}(t))\nonumber\\
&\lambda\left[ +  Q_{j,\cj}^x+ \sum_{i=1}^nQ_{j,\cj}^{y_i} + \frac{d}{\lambda}\sum_{i=1}^nC_i\frac{u_i}{\beta_i}Q_{j,\cj}^{v_i}+ \sum_{i=1}^n\frac{Q_{j,\cj}^{z_i}}{I_0^i}\right]-\left[dx+\frac{\lambda^2 Q_{j,\cj}^x}{dx}\right]-\nonumber\\
&-\lambda\sum_{i=1}^n\frac{\beta_i}{a_i}Q_{j,\cj}^{y_i}\frac{xv_i}{y_i} - dk\sum_{i=1}^n\frac{C_i}{\beta_i}Q_{j,\cj}^{v_i}\frac{y_i}{v_i}
 + \sum_{i=1}^ny_i\left[kC_i-a_i-a_iQ_{j,\cj}^{z_i}\right] +\label{master}\\
& \frac{\lambda}{d}Q_{j,\cj}^x\sum_{i=1}^nv_i\beta_i -\sum_{i=1}^nC_iu_iv_i
 + p\lambda\sum_{i=1}^na_iz_i\left[Q_{j,\cj}^{y_i}-\frac{1}{I_0^i}\right].\nonumber
\end{align}

For the first two equilibria, we consider the Lyapunov function \eqref{master}, with $C_i=a_i/k$.
Then, on using the structure of equilibria of \eqref{our:antigenic}, we may write $\dot{V}$ as follows:
\begin{align*}
 \dot{V}=& \lambda\left[1+Q^x_{j,\cj} + \sum_{i=1}^nQ^{y_i}_{j,\cj}+\sum_{i=1}^n\frac{Q^{v_i}_{j,\cj}}{R_0^i}+\sum_{i=1}^n\frac{Q^{z_i}_{j,\cj}}{I_0^i}\right]-\left[dx+\frac{\lambda^2 Q^x_{j,\cj}}{dx}\right]-\\
&\quad -\lambda\sum_{i=1}^n\frac{\beta_i}{a_i} Q^{y_i}_{j,\cj}\frac{xv_i}{y_i}-
d\sum_{i=1}^n\frac{a_i}{\beta_i}Q^{v_i}_{j,\cj}\frac{y_i}{v_i}-\sum_{i=1}^na_iQ^{z_i}_{j,\cj}y_i +\\
&\quad + 
\frac{\lambda}{d}\sum_{i=1}^n\beta_iv_i\left[Q^x_{j,\cj}-\frac{1}{R_0^i}\right]+p\lambda\sum_{i=1}^n\frac{z_i}{a_i}\left[Q^{y_i}_{j,\cj}-\frac{1}{I_0^i}\right].
\end{align*}
For $X_{0,\emptyset}$, we find, using Lemma~\ref{lem:eq}, that 
\[
 \dot{V}=2\lambda-\left[dx+\frac{\lambda^2}{dx}\right]+\frac{\lambda}{d}\sum_{i=1}^n\beta_iv_i\left[1-\frac{1}{R_0^i}\right]-p\lambda\sum_{i=1}^n\frac{z_i}{a_iI_0^i}.
\]
Since $R_0^i\leq1$, for $i=1,\dots,n$, and
\[
 dx+\frac{\lambda^2}{dx}\geq 2\lambda,
\]
we have that $\dot{V}<0$ in $\O^+$. Hence, that $X_{0,\emptyset}$ is globally asymptotically stable in this case.

Now, suppose that $1<R_0^1$, and that $P_0^1\leq1$. 
In this case, Lemma~\ref{lem:eq} yields that 
\begin{align*}
 \dot{V}=& \lambda\left[3\left(-\frac{1}{R_0^1}\right)+\frac{2}{R_0^1} \right]-\left[dx+\frac{\lambda^2}{R_0^1dx}\right]-
\frac{\lambda}{a_1}\beta_1 \left(1-\frac{1}{R_0^1}\right)\frac{xv_1}{y_1}-
\frac{a_1d}{\beta_1}(R_0^1-1)\frac{y_1}{v_1} +\\
&\quad + 
\frac{\lambda}{d}\sum_{i=1}^n\beta_iv_i\left[\frac{1}{R_0^1}-\frac{1}{R_0^i}\right]+p\lambda\frac{z_1}{a_1}\left[1-\frac{1}{R_0^1}-\frac{1}{I_0^1}\right]-p\lambda\sum_{i=2}^n\frac{z_i}{a_iI_0^i}.
\end{align*}
The last term in $\dot{V}$ is clearly negative. We also observe that, since $R_0^{i}<R_0^{1}$, for $1 < i\leq n$, we have that 
\[
\sum_{i=1}^n\beta_iv_i\left[\frac{1}{R_0^1}-\frac{1}{R_0^i}\right]<0.
\]
Also, since $P_0^1\leq1$, then we have that $R_0^1 \leq 1+ R_0^1/I_0^1$. Therefore, the last three terms in the expression for $\dot{V}$ are negative.

For the remaining terms, let us write
\[
 \frac{\lambda^2}{R_0^1}=\left(\frac{\lambda}{R_0^1}\right)^2 + \left(1-\frac{1}{R_0^1}\right)\frac{\lambda^2}{R_0^1}.
\]
Then we have that
\[
 dx+\left(\frac{\lambda^2}{R_0^1dx}\right)^2\geq 2\frac{\lambda}{R_0^1},
\]
and that
\[
.\frac{\lambda^2}{R_0^1dx}\left(1-\frac{1}{R_0^1}\right) +\lambda\frac{\beta_1}{a_1}\left(1-\frac{1}{R_0^1}\right)\frac{xv_1}{y_1}+\frac{da_1}{\beta_1}(R_0^1-1)\frac{y_1}{v_1}\geq 3\lambda\left(1-\frac{1}{R_0^1}\right).
\]
Thus $\dot{V}<0$ in $\O^+$, and hence we have that $X_{1,\emptyset}$ is a globally asymptotically stable equilibrium.

Finally, if $P_0^1>1$, then let $\cj$ be the maximal antigenic set. First, we  assume that $\cj$ is stable antigenic and show that $X_{0,\cj}$ is globally asymptotically stable. In this case, we use the Lyapunov function \eqref{master}, with $C_i=x^*\beta_i/u_i$.

Using Lemma~\ref{lem:eq}, this can be further recast as $\dot{V}=\dot{V_1}+\dot{V_2}$, where
\begin{align*}
 \dot{V_1}=&\lambda\left[3 - \frac{1}{1+\bro}\right]-\left[dx+\frac{\lambda^2 }{dx\left(1+\bro\right)}\right]
- \lambda\sum_{i\in\cj}\frac{\beta_i}{I_0^ia_i} \frac{xv_i}{y_i}
-\frac{\lambda k}{1+\bro}\sum_{i\in\cj} \frac{R_0^i}{u_iI_0^i}\frac{y_i}{v_i} ;\\
\dot{V_2}=&\sum_{i\not\in\cj}a_iy_i\left[\frac{R_0^i}{1+\bro}-1\right] -p\lambda\sum_{i\not\in\cj}\frac{z_i}{a_iI_0^i}.
\end{align*}
We treat $\dot{V_2}$ first. The last term is clearly negative. Also, since for $i\not\in\cj$, we have that 
\[
 \frac{R_0^i}{1+\bro}<1\quad\text{and thus that}\quad \sum_{i\not\in\cj}y_i\left[\frac{R_0^i}{1+\bro}-1\right]<0.
\]
Therefore, $\dot{V_2}<0$, when $\cj\not=\cn$.

Let
\[
 \eta=\frac{\bro}{1+\bro}\quad\text{and}\quad
\eta_i=\frac{\frac{R_0^i}{I_0^i}}{1+\bro},\quad i\in\cj.
\]
Then, we may write $\dot{V_1}$ as
\begin{align*}
 \dot{V_1}&=\lambda\left[3 - \frac{1}{1+\frac{\bro}{I_0}}\right]-\left[dx+\frac{ \lambda^2}{dx\left(1+\bro\right)^2}\right]
-\\ 
&\qquad - \sum_{i\in\cj}\frac{\lambda^2\eta_i}{dx\left(1+\bro\right)}
- \lambda\sum_{i\in\cj}\frac{\beta_i}{I_0^ia_i} \frac{xv_i}{y_i}
-\frac{\lambda k}{1+\bro}\sum_{i\in\cj} \frac{R_0^i}{u_iI_0^i}\frac{y_i}{v_i}
%- \sum_{i\in\cj}y_i^*\beta_i \frac{xv_i}{y_i}
%-\frac{x^*k}{u}\sum_{i\in\cj}\beta_iv_i^*\frac{y_i}{v_i}.
\end{align*}
For each $i\in\cj$, we have
\[
 -\frac{\lambda^2\frac{R_0^i}{I_0^i}}{dx\left(1+\bro\right)^2} 
- \lambda\frac{\beta_i}{I_0^ia_i} \frac{xv_i}{y_i}
-\frac{\lambda k}{1+\bro}\frac{R_0^i}{u_iI_0^i}\frac{y_i}{v_i}\leq
-3\lambda\left(\frac{\frac{R_0^i}{I_0^i}}{1+\bro}\right),
\]

and that
\[
 dx+\frac{\lambda^2}{dx\left(1+\bro\right)^2}>2\lambda\frac{1}{1+\bro}.
\]
After combining these estimates and summing for $i\in\cj$, we get that $\dot{V_1}\leq0$
and thus we have the result, If $\cj$ is a proper subset of $\cn$. In the case that $\cj=\cn$, we have $\dot{V}\leq0$, with equality ocurring only when
\[
 x=\frac{\lambda}{d}Q_{0,\cj}^x\quad\text{and}\quad \frac{v_i}{y_i}=\frac{k_i}{u_i}.
\]
Inasmuch this plane is not invariant by the corresponding flow---other than the point $X_{0,\cj}$---we have global stability as a consequence of LaSalle's theorem ~\cite{LaSalle:1964}.
For the fourth point, we use a mix of the two Lyapunov functions above, namely:
 \begin{align*}
 V(x,\mathbf{y},\mathbf{v},\mathbf{z})=&x-x^*\ln(x/x^*) + \sum_{i=1}^n\left(y_i-y_i^*\ln(y_i/y_i^*)\right) + \\ 
&\qquad +x^*\sum_{\sti{i=1}{i\not=j'}}^n\frac{\beta_i}{u_i}\left(v_i-v_i^*\ln(v_i/v_i^*)\right)
 + p\sum_{i=1}^n\frac{1}{c_i}\left(z_i-z_i^*\ln(z_i/z^*)\right)+\\
&\qquad + \frac{a_{j'}}{k}\left(v_{j'}-v_{j'}^*\ln(v_{j'}/v_{j'}^*)\right).
\end{align*}

Computing $\dot{V}$ and using the uniform notation, we find that
\begin{align*}
 \dot{V}=&\lambda\left[1 + Q_{j',\cj}^x+\sum_{i=1}^nQ_{j',\cj}^{y_i} + Q_{j',\cj}^x\sum_{\sti{i=1}{i\not=j'}}^nQ_{j',\cj}^{v_i} + \sum_{i=1}^n \frac{Q_{j',\cj}^{z_i}}{I_0^i} + \frac{Q^{v_{j'}}_{j',\cj}}{R_0^{j'}} \right]-\left[dx+\frac{\lambda^2 Q_{j',\cj}^x}{dx}\right]-\\
&\qquad - \lambda\sum_{i=1}^n\frac{\beta_i Q_{j',\cj}^{y_i}}{a_i} \frac{xv_i}{y_i}
-\lambda kQ_{j',\cj}^x\sum_{\sti{i=1}{i\not=j'}}^n\frac{Q_{j',\cj}^{v_i}}{u_i}\frac{y_i}{v_i} + \sum_{\sti{i=1}{i\not=j'}}^ny_i\left[\frac{\beta_i \lambda k}{du_i}Q_{j',\cj}^x-a_i-a_iQ_{j',\cj}^{z_i}\right]+\\ &\qquad +p\lambda\sum_{i=1}^n\frac{z_i}{a_i}\left[Q_{j',\cj}^{y_i}-\frac{1}{I_0^i}\right] -
\frac{da_{j'}}{\beta_{j'}}Q^{v_{j'}}_{j',\cj}\frac{y_{j'}}{v_{j'}} - a_{j'}Q^{z_{j'}}_{j',\cj}y_{j'} +
\frac{\lambda\beta_{j'}}{d}v_{j'}\left[Q^x_{j',\cj}-\frac{1}{R_0^{j'}}\right]
\end{align*}

On using Lemma~\ref{lem:eq}, and that
\[
1=\left(1-\frac{1}{R_0^{j'}}-\frac{\bro}{R_0^{j'}}\right)+\frac{1}{R_0^{j'}}+\frac{\bro}{R_0^{j'}},
\]
where each term in the sum is positive, but smaller than one, we rewrite it as $ \dot{V}=\dot{V_1}+\dot{V_2} + \dot{V_3} + \dot{V_4}$, where
\begin{align*}
\dot{V_1}&=\left(1-\frac{1}{R_0^{j'}}-\frac{\bro}{R_0^{j'}}\right)\left[3\lambda-\frac{\lambda^2}{dxR_0^{j'}}-\frac{\lambda\beta_{j'}}{a_{j'}}\frac{xv_{j'}}{y_{j'}}-\frac{da_{j'}R_0^{j'}}{\beta_{j'}}\frac{y_{j'}}{v_{j'}}\right],\\
\dot{V_2}&=\frac{2\lambda}{R_0^{j'}}-\left[dx+\frac{\lambda^2}{dx(R_0^{j'})^2}\right],\\
\dot{V_3}&=3\lambda\frac{\bro}{R_0^{j'}}-\frac{\lambda^2\bro}{dx(R_0^{j'})^2}-\lambda\sum_{i\in\cj}\frac{\beta_i}{a_iI_0^i}\frac{xv_i}{y_i}-\frac{\lambda k}{R_0^{j'}}\sum_{i\in\cj}\frac{R_0^i}{u_iI_0^i}\frac{y_i}{v_i},\\
\dot{V_4}&=\frac{p\lambda z_{j'}}{a_{j'}}\left(1-\frac{1}{R_0^{j'}}-\frac{\bro}{R_0^{j'}}-\frac{1}{I_0^{j'}}\right).
\end{align*}
The terms $\dot{V_1}$, $\dot{V_2}$ and $\dot{V_3}$ can be treated similarly as in the previous equilibria and are all nonpositive in the interior.

As for $\dot{V_4}$, first we observe that
\[
 1-\frac{1}{R_0^{j'}}-\frac{\bro}{R_0^{j'}}-\frac{1}{I_0^{j'}}=\frac{1}{R_0^{j'}}\left(R_0^{j'}-1-\bro-\frac{R_0^{j'}}{I_0^{j'}}\right)
\]
if $j'\not\in\cs$, then we that
\[
 R_0^{j'}-1-\bro-\frac{R_0^{j'}}{I_0^{j'}}<R_0^{j'}-1-\frac{R_0^{j'}}{I_0^{j'}}\leq0.
\]
If $j'\in\cs$, then let $\cj'=\cj\cup\{j'\}$. Then we have that
\[
 1-\frac{1}{R_0^{j'}}-\frac{\bro}{R_0^{j'}}-\frac{1}{I_0^{j'}}=\frac{1}{R_0^{j'}}\left(R_0^{j'}-1 -\rho_0^{\cj'}\right).
\]
Since $\cj$ is the maximal antigenic set, we must have that 
\[
 R_0^{j'}-1 -\rho_0^{\cj'}\leq0.
\]
\end{proof}
\subsection{Nonunique Fitness}

Given the non-generic nature of this case,  we shall only briefly discuss the stability when some of the strains have the same fitness, i.e., there exists at lest one  index set $\Gamma$, such that $R_0^i=R_0^j$, for $i,j\in\Gamma$. Notice that, in this case, we have non-isolated equilibria.
 
We start by observing that the computation with the Lyapunov function for the equilibrium $X_{0,\emptyset}$ does not depend on the uniqueness of fitness. Hence we have
\begin{cor}
 If $R_0^i\leq1$, for $i=1,\ldots,n$, then $X_{0,\emptyset}$ is a globally asymptotically stable equilibrium.
\end{cor}
The case $I_0^i=I_0$, $1<R_0^1$ and $P_0^1<1$ can also be partially treated:
\begin{prop}
\label{prop:ls}
 Let $\Gamma$ be the set of indices $i\in\cn$ such that $R_0^i=R_0^n$. Let $E_\Gamma$ be the set satisfying
\[
x^*=\frac{\lambda}{dR_0^n},\quad y_j=v_j=0,\quad j\not\in\Gamma,\quad \sum_{j\in\Gamma}\beta_jv_j=\frac{dx^*-\lambda}{x^*}, \quad v_j\geq0,
\]
\[
v_i=\frac{k}{u}y_i,\quad\text{and}\quad z_i=0, \quad i\in\cn.
\]
If $\mathbf{x}(t,\mathbf{x}_0)$ is a solution of \eqref{nb:antigenic}, with initial condition $\mathbf{x}_0$, then 
\[
 \mathbf{x}(t)\to E_\Gamma,\quad\text{as}\quad t\to\infty.
\]
\end{prop}

\begin{proof}
Let us denote the omega set of $\mathbf{x}(t,\mathbf{x}_0)$ by $\Omega(\mathbf{x_0})$. As shown in \cite{Pastore:2005}, the solutions to \eqref{nb:antigenic} are bounded in $\mathbb{R}_+^{3n+1}$.
Hence, $\Omega(\mathbf{x_0})$ is compact. 

Using the same Lyapunov function for $X_{1,\emptyset}$ as in Section 3.2, we find that
\begin{align*}
 \dot{V}=& \lambda\left[3-\frac{1}{R_0^1} \right]-\left[dx+\frac{\lambda^2}{R_0^1dx}\right]-
\frac{\lambda}{a}\beta_1 \left(1-\frac{1}{R_0^1}\right)\frac{xv_1}{y_1}-
\frac{ad}{\beta_1}(R_0^1-1)\frac{y_1}{v_1} +\\
&\quad + 
\frac{\lambda}{d}\sum_{i\not\in\Gamma}\beta_iv_i\left[\frac{1}{R_0^1}-\frac{1}{R_0^i}\right]+\frac{p\lambda}{a}\sum_{i=1}^nz_i\left[1-\frac{1}{R_0^i}-\frac{1}{I_0}\right].
\end{align*}
The same calculations in section 3.2 shows that $\dot{V}\leq0$. However, notice that $\dot{V}=0$ in $E_\Gamma$. If $s\in\mathbb{R}^N$ and $S\subset\mathbb{R}^N$ is closed, then let
\[
 d(s,S)=\min_{s'\in S}\|s-s'\|.
\]
LaSalle's invariant principle then yields that,
\[
 \lim_{t\to\infty}d(\mathbf{x}(t),E_\Gamma)=0.
\]

\end{proof}

\begin{cor}
 If $R_0^1>R_0^i$ for $i=2,\ldots,n$, then $X_{1,\emptyset}$ is a globally asymptotically stable equilibrium, when $1<R_0^1\leq1+R_0^1/I_0$.
\end{cor}
We have performed numerical calculations of \eqref{nb:antigenic}, using a high order Runge-Kutta method, which suggest that  in the case treated by proposition~1, a solution of System~\eqref{nb:antigenic} will converge to a unique equilibrium point in $E_\Gamma$, that depends  only on the initial condition.

Finally, when the viable set of strains is not the full antigenic variation, we have
\begin{prop}
 If $I_0^i=I_0$, $P_0^1>1$. Assume that  $\cj\not=\ci$ is a stable antigenic set, then we have that $X_{0,\cj}$ is a globally asymptotic stable equilibrium.
\end{prop}

\begin{proof}
 In this case, the estimate  $\dot{V}_1\leq0$ remains valid. Moreover, if $\cj\not=\ci$, then we must have $\dot{V}_2<0$, which yields the result.
\end{proof}

\section{Conclusions} \label{conclusions}
In this work, we have performed a thorough study of System~\eqref{nb:antigenic}. 
As a preliminary result, we have shown that, when the both the virus production rate and the CTL interaction rate are nonzero for all strains, then System~\eqref{nb:antigenic} is dynamically equivalent to System~\eqref{our:antigenic}, which has strain-independent virus production  and CTL interaction rates. In particular, the precise nonzero values of the CTL interaction rates are completely irrelevant to the dynamical behaviour of the system. This seems to suggest that a more refined model which is able to capture this difference is needed.

We have also identified all the $2^{n-1}(2+n)$ equilibria of \eqref{our:antigenic} in Lemma~\ref{lem:eq}. When $n$ is large this can be quite a large number. Nevertheless, under the hypothesis of unique fitness, we were able to show that only four of them are dynamically relevant. More precisely, we assume that the virus basic reproduction rates, $R_0^i$ are distinct and that the CTL reproduction numbers, $P_0^i$ have the same ordering as the $R_0^i$. This last condition is automatically satisfied if $I_0^i=I_0$ for all $i$. In this case, Theorem~\ref{thm:mutation} shows that, if the largest reproduction number, $R_0^1$, is smaller than one, the the disease-free equilibrium---$X_{0,\emptyset}$ in the notation of Lemma~\ref{lem:eq}---is globally asymptotic stable. In this case, no strain is viable and the infection dies out. On the other hand, if $R_0^1>1$, but $P_0^1<1$, then only the first strain survives, and the infection persists. If $P_0^1>1$, then we have that the two outcomes are possible: either a (unique) stable antigenic set exists and then $X_{0,\cj}$ is globally asymptotic stable. In this case, the set $\cj$ determines the antigenic diversity.   In other words, a strong immune response generates a larger antigenic variation. Alternatively, there exits  a pair $(j',\cj)$ such that the point $X_{j',\cj}$ is globally asymptotic stable. In this case, the strain with the weakest fitness  will not actually trigger the CTL response at all in the long run. 
In the case of absence of antigenic variation, i.e. System~\eqref{nb:immune}, then only the first outcome is possible. We were unable to interpret in a biological sense the combinatorial conditions of existence of a stable antigenic set, and we believe that this should be addressed in the future.
The results presented in Theorem~\ref{thm:mutation} show rigorously some of the inferences that have already been made in \citeasnoun{Nowak:Bangham:1996} based on extensive simulations of System~\eqref{nb:antigenic}.

We have also shown some results for very special  cases in which the $R_0^i$s are not distinct. In these cases, the equilibria is not isolated and this complicates the matters further. Also, we have not addressed that case when the set of strong responders is not consistent, and this might also merit further study in the future.

\bibliographystyle{harvard}

\bibliography{hiv,bbb}

\end{document}